\documentclass[12pt]{amsart}
\usepackage{a4wide}
\usepackage{amssymb}
\usepackage{wrapfig,epsfig}
\usepackage{color}

\newtheorem{definition}{Definition}[section]
\newtheorem{lemma}{Lemma}
\newtheorem{proposition}{Proposition}
\newtheorem{theorem}{Theorem}

\makeatletter
\@addtoreset{equation}{section}

\newcommand{\bes}{\begin{displaymath}}
\newcommand{\ees}{\end{displaymath}}
\newcommand{\be}{\begin{equation}}
\newcommand{\ee}{\end{equation}}
\newcommand{\ba}{\begin{eqnarray}}
\newcommand{\ea}{\end{eqnarray}}
\newcommand{\bas}{\begin{eqnarray*}}
\newcommand{\eas}{\end{eqnarray*}}
\newcommand{\@Bbb}[1]{\ensuremath{\Bbb #1}}

\newcommand{\B}{{\@Bbb B}}
\newcommand{\C}{{\@Bbb C}}
\newcommand{\E}{{\@Bbb E}}
\newcommand{\F}{{\@Bbb F}}
\newcommand{\G}{{\@Bbb G}}
\renewcommand{\P}{{\@Bbb P}}

\newcommand{\Q}{{\@Bbb Q}}
\newcommand{\bQ}{{\@Bbb Q}}
\newcommand{\N}{{\@Bbb N}}
\newcommand{\R}{{\@Bbb R}}
\newcommand{\T}{{\@Bbb T}}
\newcommand{\bbR}{{\@Bbb R}}
\newcommand{\W}{{\@Bbb W}}
\newcommand{\Z}{{\@Bbb Z}}
\newcommand{\bbZ}{{\@Bbb Z}}

\newcommand{\bone}{{\bf 1}}

\newcommand{\om}{\omega}

\newcommand{\ep}{\varepsilon}

\newcommand{\@s}[1]{\ensuremath{\mathcal #1}}
\newcommand{\cA}{\@s A}
\newcommand{\cB}{\@s B}
\newcommand{\cC}{\@s C}
\newcommand{\cD}{\@s D}
\newcommand{\cE}{\@s E}
\newcommand{\cF}{\@s F}
\newcommand{\cG}{\@s G}
\newcommand{\cH}{\@s H}
\newcommand{\cI}{\@s I}
\newcommand{\cJ}{\@s J}
\newcommand{\cal}{\mathcal}
\newcommand{\cK}{\@s K}
\newcommand{\cL}{\@s L}
\newcommand{\cN}{\@s N}
\newcommand{\cM}{\@s M}
\newcommand{\cO}{\@s O}
\newcommand{\cP}{\@s P}
\newcommand{\cQ}{\@s Q}
\newcommand{\cR}{\@s R}
\newcommand{\cS}{\@s S}
\newcommand{\cT}{\@s T}
\newcommand{\cU}{\@s U}
\newcommand{\cV}{\@s V}
\newcommand{\cW}{\@s W}
\newcommand{\cX}{\@s X}
\newcommand{\cY}{\@s Y}
\newcommand{\cZ}{\@s Z}

\def\REQ#1{{\rm (\ref{#1})}}
\def\d{{\rm d}}

\newcommand{\@bm}[1]{\ensuremath{\mathbf #1}}
\newcommand{\bma}{\@bm a}\newcommand{\bmA}{\@bm A}
\newcommand{\bmb}{\@bm b}\newcommand{\bmB}{\@bm B}
\newcommand{\bmc}{\@bm c}\newcommand{\bmC}{\@bm C}
\newcommand{\bmd}{\@bm d}\newcommand{\bmD}{\@bm D}
\newcommand{\bme}{\@bm e}
\newcommand{\bmf}{\@bm f}\newcommand{\bmF}{\@bm F}
\newcommand{\bmg}{\@bm g}\newcommand{\bmG}{\@bm G}
\newcommand{\bmh}{\@bm h}\newcommand{\bmH}{\@bm H}
\newcommand{\bmi}{\@bm i}\newcommand{\bmI}{\@bm I}
\newcommand{\bmj}{\@bm j}
\newcommand{\bmk}{\@bm k}\newcommand{\bmK}{\@bm K}
\newcommand{\bml}{\@bm l}
\newcommand{\bmm}{\@bm m}\newcommand{\bmM}{\@bm M}
\newcommand{\bmn}{\@bm n}
\newcommand{\bmo}{\@bm o}
\newcommand{\bmp}{\@bm p}
\newcommand{\bmq}{\@bm q}\newcommand{\bmQ}{\@bm Q}
\newcommand{\bmr}{\@bm r}
\newcommand{\bms}{\@bm s}\newcommand{\bmS}{\@bm S}
\newcommand{\bmt}{\@bm t}
\newcommand{\bmu}{\@bm u}\newcommand{\bmU}{\@bm U}
\newcommand{\bmw}{\@bm w}\newcommand{\bmW}{\@bm W}
\newcommand{\bmv}{\@bm v}\newcommand{\bmV}{\@bm V}
\newcommand{\bmx}{\@bm x}\newcommand{\bmX}{\@bm X}\newcommand{\bx}{\@bm x}
\newcommand{\bmy}{\@bm y}\newcommand{\bmY}{\@bm Y}\newcommand{\by}{\@bm y}
\newcommand{\bmz}{\@bm z}\newcommand{\bmZ}{\@bm Z}

\newcommand{\bmzero}{\@bm 0}

\newcommand{\Tr}{\mathop{\rm Tr}}

\newcommand{\@g}[1]{\ensuremath{\mathfrak #1}}
\newcommand{\gA}{\@g A}
\newcommand{\gD}{\@g D}
\newcommand{\gJ}{\@g J}
\newcommand{\gF}{\@g F}
\newcommand{\gM}{\@g M}
\newcommand{\gR}{\@g R}

\newcommand{\D}{\mathop{\mbox{D}}}

\newcommand{\commentout}[1]{{}}

\makeatother
\begin{document}
\title[The stability of some stochastic processes ]{The stability of some stochastic processes}
\author{H. Bessaih}
\address{Department of Mathematics, University of Wyoming, Laramie, Wyoming, U.S.A.}
\email{Bessaih@uwyo.edu}
\author{R. Kapica}
\address{Department of Mathematics, University of Silesia, Bankowa 14, 40-007 Katowice, Poland}
\email{rkapica@math.us.edu.pl}
\author{T. Szarek}
\address{Institute of Mathematics, University of Gda\'nsk, Wita Stwosza
  57, 80-952 Gda\'nsk, Poland}
\email{szarek@intertele.pl}

\begin{abstract} We formulate and prove a new criterion for stability of e--processes. It says that any e--process which is averagely bounded and concentrating is asymptotically stable. In the second part, we show how this general result applies to some shell models (the Goy and the Sabra model).
Indeed, we manage to prove that the processes corresponding to these models satisfy the e--process property. They are also averagely bounded and concentrating. Consequently, their stability follows.
\end{abstract}


\thanks{Tomasz Szarek has been supported by Polish Ministry of Science and
Higher Education  Grants N N201 419139. }

\date\today
\maketitle

\section{Introduction}

The paper is aimed at formulating and proving a new criterion for asymptotic stability of Markov processes. We are concerned with processes satisfying the so-called e--property (see \cite{ Ku-Sh, Lasota-Szarek}). The e--property is a generalization of nonexpansiveness and it allows us to overcome some restrictiveness of the strong Feller property when applied to some SPDE's (see for instance \cite{KPS10}). Recall that a process is called nonexpansive if the Markov semigroup (acting on measures) corresponding to the process is nonexpansive with respect to the Wasserstein metric. A. Lasota and J. Yorke (see \cite{LasYorke}) found a very elegant and concise condition assuring the existence and uniqueness of an invariant measure for nonexpansive Markov chains. It says that the considered chain is concentrating at some point, i.e.  the chain remains in arbitrary small neighborhoods of some fixed point with positive probability independent of an initial point. The proof based on the lower bound technique was developed in \cite{LasYorke, LasMackey}. The result mentioned above was proved in locally compact spaces, subsequently the proof was also given in Polish spaces \cite{Szar} and finally the result was extended to Markov processes \cite{Trap}. Here, we prove that if a Markov process is averagely bounded and with positive probability enters into any neighborhood of a fixed point, then this process is asymptotically stable. In particular, it admits a unique invariant measure.

We strongly believe that the criterion proved in the first part of the paper may be useful in the theory of SPDE's in particular with L\'evy noise.
Here, it is applied to some stochastic shell models of turbulence, the GOY and Sabra model. These are very popular examples of simplified phenomenological models of turbulence. Although they are not based on conservations laws, they capture some essential statistical properties and features of turbulent flows, like the energy and the enstrophy cascade and the power law decay of the structure functions in some range of wave numbers, the inertial range. We refer the reader to \cite{OY89}, \cite{ABCPV01}, \cite{Biferale03},  \cite{Frisch} and \cite{Gallavotti} and the references therein and to \cite{CLT06}, \cite{FlaCIME}  and \cite{Hakima} for some rigorous results. We are interested in a noise where only finitely many modes are nonzero and then we prove the e-process property. It is possible that the simillar results may be obtained using coupling methods (see for instance \cite{Hai, Ku-Sh, BKL}). However we provide this application for its simplicity. Indeed, the proof of the e--property makes use of the Malliavin calculus developed in \cite{HM06}. Boundedness and concentrating property, in turn, easily follow from standard estimates for shell models. Their proofs are rather straightforward. The main result of this part of our paper answers to the conjecture posed by Barbato {\it et al.} (see \cite{Hakima}) who anticipated that in the case when the number of modes to which we add the noise is large enough, it would be possible to prove the uniqueness of an invariant measure.

The paper is organized as follows. In section 2, we introduce the concepts of e-property, averagely bounded and concentrating at a point. We also prove (Proposition 1) the main result about asymptotic stability for Markov processes. In Section 3, we introduce the GOY and Sabra models and give general results about their well posedness. In Section 4, we apply the results of Section 2 to the shell models and prove the e-property, the average boundedness and the concentrating property and state our main result for the uniqueness of the invariant measure for the stochastic shell model with a degenerate noise.

\section{Criterion on Stability}
Let $(X, \rho)$ be a Polish space.
By $B_b(X)$ we denote the space of all bounded Borel--measurable functions equipped with the supremum norm.
Let $(P_t)_{t\ge 0}$ be the \emph{Markovian semigroup} defined on $B_b(X)$.
For each $t\ge 0$ we have $P_t\bone_X=\bone_X$ and $P_t\psi\ge 0$ if $\psi\ge 0$.
Throughout this paper we shall assume that the semigroup is \emph{Feller}, i.e.
$P_t(C_b(X))\subset C_b(X)$ for all $t>0$. Here and in the sequel $C_b(X)$ is the subspace of all
bounded continuous functions with the supremum norm $\|\cdot\|_{\infty}$. By $L_b(X)$ we will denote the subspace of all bounded Lipschitz functions.
We shall also assume that $(P_t)_{t\ge 0}$ is \emph{stochastically continuous}, which implies that $\lim_{t\to 0^+} P_t \psi(x)=\psi(x)$ for all $x\in X$ and $\psi\in C_b(X)$.

Let ${\cal M}_1$
stand for the space
of all Borel probability measures on $X$.
Denote by $\mathcal M_{1}^{W}$, $W\subset X$, the subspace of all Borel probability measures supported in $W$,
i.e. $\{x\in X: \mu(B(x,r))>0\text{ for any }r>0\}\subset W$, where $B(x, r)$ denotes the ball in $X$ with center at $x$ and radius $r$.
For $\varphi\in B_b(X)$ and $\mu\in\mathcal M_1$ we will use the notation
$\langle \varphi,\mu\rangle =\int_X\varphi(x)\mu(\d x)$.
Recall that the {\it total variation norm} of a finite signed measure $\mu\in\cal M_1-\cal M_1$ is given by
$\|\mu\|_{TV}=\mu^{+}(X)+\mu^{-}(X)$, where $\mu=\mu^{+}-\mu^{-}$ is the Jordan decomposition of $\mu$.

We say that $\mu_* \in \mathcal M_1$ is
\emph{invariant} for $(P_t)_{t\ge 0}$ if
$
\langle P_t\psi, \mu_*\rangle = \langle\psi, \mu_*\rangle
$
for every $\psi\in B_b(X)$ and $t\ge 0.$
Alternatively, we can say that $P_t^*\mu_*=\mu_*$ for all $t\ge 0$, where
$(P_t^*)_{t\ge 0}$ denotes the semigroup dual to $(P_t)_{t\ge 0}$, i.e. for a given  Borel measure $\mu$, Borel
subset $A$ of $X$, and $t\ge0$ we set
$$
P_t^*\mu (A):=\langle P_t \bone_A, \mu \rangle.
$$

A semigroup $(P_{t})_{t\ge 0}$ is said to be \emph{asymptotically stable} if there exists an invariant measure $\mu_* \in \mathcal M_1$ such that $P_{t}^{*}\mu$ converges weakly to $\mu_{*}$ as $t\to+\infty$
for every $\mu \in \mathcal M_1$. Obviously $\mu_{*}$ is unique.

\begin{definition}
We say that a  semigroup $(P_t)_{t\ge 0}$ has
\emph{the e--property} if the family of functions $(P_t\psi)_{t\ge 0}$  is
equicontinuous at every point $x$ of $X$ for any bounded and
Lipschitz function $\psi$, i.e.
$$
\forall\,\psi\in L_b(X),\,x\in X,\,\ep>0\,\exists\,\delta>0\,\forall\,z\in B(x,\delta),\,t\ge0:\,|P_t\psi(x)-P_t\psi(z)|<\ep.
$$
\end{definition}
\noindent {\bf Remark.} One can show (see \cite{KPS10}) that to obtain the e--property in the case when $X$ is a Hilbert space, it is enough to verify
the above condition for every function with bounded Fr\'echet derivative.

\begin{definition} A semigroup $(P_{t})_{t\ge 0}$ is called \emph{averagely bounded} if for any $\ep>0$ and bounded set $A\subset X$ there is a bounded Borel set $B\subset X$ such that
$$
\limsup_{T\to\infty}\frac{1}{T}\int_{0}^{T}P_{s}^{*}\mu(B)\d s>1-\ep\qquad\text{for $\mu\in\mathcal M_{1}^{A}$.}
$$
\end{definition}

\begin{definition} A semigroup $(P_{t})_{t\ge 0}$ is \emph{concentrating} at $z$ if for any $\ep>0$ and bounded set $A\subset X$ there exists $\alpha>0$ such that for any two measures $\mu_{1}, \mu_{2}\in\mathcal M_{1}^A$ holds
$$
P_{t}^{*}\mu_{i}(B(z, \ep))\ge\alpha\,\,\text{for $i=1, 2$ and some $t>0$}.
$$
\end{definition}

\begin{proposition}\label{prop1}
Let $(P_{t})_{t\ge 0}$ be averagely bounded and concentrating at some $z\in X$. If $(P_{t})_{t\ge 0}$ satisfies the e--property, then for any $\varphi\in L_{b}(X)$ and $\mu_{1}, \mu_{2}\in\mathcal M_{1}$ we have
\begin{equation}\label{eq1}
\lim_{t\to\infty}|\langle \varphi, P_{t}^{*}\mu_{1}\rangle-\langle \varphi, P_{t}^{*}\mu_{2}\rangle|=0.
\end{equation}
\end{proposition}

\begin{proof} First observe that to finish the proof it is enough to show that condition (\ref{eq1}) holds for arbitrary Borel probability measures with bounded support.  Indeed, the set of all probability measures with bounded support is dense in the space $(\mathcal M_{1}, \|\cdot\|_{TV})$. Moreover, $P_{t}^{*}$, $t\ge 0$, is nonexpansive with respect to the total variation norm.

Fix $\varphi\in {L}_b(X)$, $x_0\in X$ and $\ep\in (0, 1/2)$. Let $\mu_{1}, \mu_{2}\in\mathcal M_{1}^{B(x_0, r_{0})}$ for some $r_{0}>0$.
Choose $\delta>0$ such that
\begin{equation}\label{nee1}
\sup_{t\ge 0}|P_{t}\varphi(x)-P_{t}\varphi(y)|<\ep/2
\end{equation}
for $x, y\in B(z, \delta)$, by the e-property.

Since $(P_{t})_{t\ge 0}$ is averagely bounded we may find $R_{0}>0$ such that
\begin{equation}\label{eq3}
\limsup_{T\to\infty}\frac{1}{T}\int_{0}^{T} P_{s}^{*}\mu(B(x_0, R_{0}))\d s>1-\ep^{2}/(4\|\varphi\|_{\infty})
\end{equation}
for any $\mu\in\mathcal M_{1}^{B(x_0, r_{0})}$.
Let $R>\max\{R_{0}, r_{0}\}$ satisfy
\begin{equation}\label{eq01}
\limsup_{T\to\infty}\frac{1}{T}\int_{0}^{T} P_{s}^{*}\mu(B(x_0, R))\d s>3/4
\end{equation}
for any $\mu\in\mathcal M_{1}^{B(x_0, R_{0})}$.
Since $(P_{t})_{t\ge 0}$ is concentrating at $z$ we may choose $\alpha>0$ such that for any $\nu_1, \nu_2\in \mathcal M_1^{B(x_0, R)}$ there exists $t>0$ and the condition
\begin{equation}\label{eq02}
P_{t}^{*}\nu_i (B(z, \delta))\ge\alpha\qquad\text{for $i=1, 2$}
\end{equation}
holds.

Set $\gamma:=\alpha\ep/2>0$. Let $k$ be the minimal integer such that $4(1-\gamma)^{k}\|\varphi\|_{\infty}\le\ep$.

We will show by induction that for every $l\le k$, $l\in\mathbb N$, there exist $t_{1},\ldots, t_{l}>0$ and
$\nu^{i}_{1},\ldots, \nu^{i}_{l}, \mu^{i}_{l}\in\mathcal M_{1}$ such that $\nu^{i}_{j}\in\mathcal M_{1}^{B(z, \delta)}$ for $j=1, \ldots, l$ and
\begin{equation}\label{eq4}
\begin{aligned}
P_{t_{1}+\cdots+t_{l}}^{*}\mu_{i}&=\gamma P_{t_{2}+\cdots+t_{l}}^{*}\nu^{i}_{1}+\gamma(1-\gamma) P_{t_{3}+\cdots+t_{l}}^{*}\nu^{i}_{2}\\
&+\cdots +\gamma(1-\gamma)^{l-1}\nu^{i}_{l} +(1-\gamma)^{l}\mu^{i}_{l}\qquad\text{for $i=1, 2$}.
\end{aligned}
\end{equation}

Indeed, let $t_{1}>0$ be such that
$$
P_{t_{1}}^{*}\mu_{i}(B(z, \delta))\ge\alpha>\gamma\qquad\text{for $i=1, 2$}.
$$
Set
\begin{equation}\label{eq5}
\nu^{i}_{1}=\frac{P_{t_{1}}^{*}\mu_{i}(\,\cdot\,\cap B(z, \delta))}{P_{t_{1}}^{*}\mu_{i}(B(z, \delta))},
\end{equation}
$$
\mu_{1}^{i}=(1-\gamma)^{-1}(P_{t_{1}}^{*}\mu_{i}-\gamma\nu_{1}^{i})\qquad\text{for $i=1, 2$}
$$
and observe that $ \mu^{i}_{1}\in\mathcal M_{1}$ and $\nu^{i}_{1}\in\mathcal M_{1}^{B(z, \delta)}$ for $i=1, 2$. Then condition (\ref{eq4}) holds for $l=1$.

Now assume that we have done it for some $l$ and $4(1-\gamma)^{l}\|\varphi\|_{\infty}>\ep$. Then there exist $s_{i}>0$ for $i=1, 2$ such that
$$
P_{t_{1}+\cdots+t_{l}+s_{i}}^{*}\mu_{i}(X\setminus B(x_0, R_{0}))<\ep^{2}/(4\|\varphi\|_{\infty})
$$
for $i=1, 2$, by (\ref{eq3}). Since $(1-\gamma)^{l}>\ep/(4\|\varphi\|_{\infty})$, from the linearity of $P_{s_{i}}^{*}$ we obtain that
$$
P_{s_{i}}^{*}\mu_{l}^{i}(B(x_0, R_{0}))>\ep\qquad\text{for $i=1, 2$}.
$$
Thus we may find two measures $\tilde{\mu}_{l}^{1}, \tilde\mu_{l}^{2}\in\mathcal M_{1}^{B(x_0, R_0)}$ such that
\begin{equation}\label{eq7}
P_{s_{i}}^{*}\mu_{l}^{i}\ge\ep\tilde{\mu}_{l}^{i}.
\end{equation}
These measures may be defined as restriction of $P_{s_{i}}^{*}\mu_{l}^{i}$ to $B(x_0, R_0)$ respectively normed (see formula (\ref{eq5})). Further, from (\ref{eq01}) it follows that
$$
\begin{aligned}
&\limsup_{T\to\infty}\frac{1}{T}\int_{0}^{T} [P_{s+s_{2}}^{*}(\tilde{\mu}^{1}_{l}/2)(B(x_0, R))+P_{s+s_{1}}^{*}(\tilde{\mu}^{2}_{l}/2)(B(x_0, R))]\d s\\
&=\limsup_{T\to\infty}\frac{1}{T}\int_{0}^{T} P_{s}^{*}(\tilde{\mu}_{l}^{1}/2+\tilde{\mu}_{l}^{2}/2)(B(x_0, R))\d s>3/4,
\end{aligned}
$$
by the fact that $\tilde{\mu}_{l}^{1}/2+\tilde{\mu}_{l}^{2}/2\in \mathcal M_1^{B(x_0, R_0)}$. Consequently, for some $s>0$ we have
$$
P_{s+s_{2}}^{*}\tilde{\mu}^{1}_{l}(B(x_0, R))\ge1/2\quad\text{and}\quad P_{s+s_{1}}^{*}\tilde{\mu}^{2}_{l}(B(x_0, R))\ge1/2.
$$
Comparing (\ref{eq7}) and the above we obtain
$$
P_{s+s_{1}+s_{2}}^{*}\mu_{l}^{i}\ge(\ep/2)\hat{\mu}_{l}^{i}
$$
for some $\hat\mu_{l}^{i}\in\mathcal M_{1}^{B(x_0, R)}$, $i=1, 2$, by argument similar to that in (\ref{eq7}). Using it once again and taking into consideration
(\ref{eq02}) we obtain that there exists $t>0$ such that
$$
P_{t+s+s_{1}+s_{2}}^{*}\mu_{l}^{i}\ge(\alpha\ep/2)\nu_{l+1}^{i}=\gamma\nu_{l+1}^{i}
$$
for some $\nu_{l+1}^{i}\in\mathcal M_{1}^{B(z, \delta)}$ for $i=1, 2$.
Therefore, setting $t_{l+1}=t+s+s_{1}+s_{2}$ we obtain
$$
\begin{aligned}
P_{t_{1}+\cdots+t_{l}+t_{l+1}}^{*}\mu_{i}&=\gamma P_{t_{2}+\cdots+t_{l+1}}^{*}\nu^{i}_{1}+\gamma(1-\gamma) P_{t_{3}+\cdots+t_{l+1}}^{*}\nu^{i}_{2}\\
&+\cdots +\gamma(1-\gamma)^{l-1}P_{t_{l+1}}^{*}\nu^{i}_{l} +\gamma(1-\gamma)^{l}\nu^{i}_{l+1}+(1-\gamma)^{l+1}\mu_{l+1}^{i},
\end{aligned}
$$
where
$$
\mu_{l+1}^{i}=(1-\gamma)^{-1}(P_{t_{l+1}}^{*}\mu_{l}^{i}-\gamma\nu_{l+1}^{i})\qquad\text{for $i=1, 2$}.
$$
This completes the proof of condition (\ref{eq4}). In turn, this and (\ref{nee1}) give for $t\ge t_{1}+\cdots+t_{k}$
$$
\begin{aligned}
|\langle \varphi, P_{t}^{*}\mu_{1}\rangle-\langle \varphi, P_{t}^{*}\mu_{2}\rangle|&=
|\langle  P_{t-(t_{1}+\cdots t_{k}}\varphi, P_{t_{1}+\cdots t_{k}}^{*}\mu_{1}\rangle-
\langle  P_{t-(t_{1}+\cdots t_{k}}\varphi, P_{t_{1}+\cdots t_{k}}^{*}\mu_{2}\rangle|\\
&\le \gamma |\langle P_{t-t_{1}}\varphi, \nu_{1}^{1}-\nu_{1}^{2}\rangle| +
\gamma(1-\gamma)| \langle P_{t-(t_{1}+t_{2})}\varphi, \nu_{2}^{1}-\nu_{2}^{2}\rangle|+\cdots\\
&+\gamma(1-\gamma)^{k-1}|\langle P_{t-(t_{1}+\cdots+t_{k})}\varphi, \nu_{k}^{1}-\nu_{k}^{2}\rangle|
+2(1-\gamma)^{k}\|\varphi\|_{\infty}\\
&\le (\gamma+\gamma(1-\gamma)+\cdots+\gamma(1-\gamma)^{k-1})\sup_{t\ge 0,\, x, y\in B(z, \delta)}|P_{t}\varphi(x)-P_{t}\varphi(y)|\\
&+\ep/2\le\ep/2 +\ep/2=\ep.
\end{aligned}
$$
Since $\ep>0$ was arbitrary, the proof is complete.
\end{proof}

\begin{proposition}\label{prop2}
Assume that there exists $z\in X$ such that for any $\ep>0$
\begin{equation}\label{inv}
\limsup_{T\to\infty}\sup_{\mu\in\mathcal M_{1}}\frac{1}{T}\int_{0}^{T}P_{s}^{*}\mu(B(z, \ep))\d s>0.
\end{equation}
If $(P_{t})_{t\ge 0}$ satisfies the e--property, then it admits an invariant measure.
\end{proposition}

\begin{proof}
Assume, contrary to our claim, that $(P_{t})_{t\ge 0}$ does not possess any invariant measure.
From Step I of Theorem 3.1 in
\cite{Lasota-Szarek} it follows that then there exists
an $\ep>0$, a sequence of compact sets $(K_i)_{i\ge 1}$, and an increasing sequence of positive reals $(q_i)_{i\ge 1}$,
$q_i\to\infty$, satisfying
$$
P_{q_i}^{*}\delta_z(K_i)\ge\ep
\qquad\text{for $i\in\mathbb N$}
$$
and
$$
\min\{\rho(x,y): x\in K_i, y\in K_j\}\ge\ep
\qquad\text{for $i\not= j$, $i, j\in\mathbb N$}.
$$

We will show that
for every open neighborhood $U$ of $z$ and every $i_0\in\mathbb N$ there exists $y\in U$ and $i\ge i_0$, $i\in\mathbb N$, such that
$$
P_{q_i}^{*}\delta_y\left( K_i^{\ep/3}\right)<\ep/2,
$$
where $K_i^{\ep/3}=\{y\in X: \inf_{v\in K_{i}} \rho (y, v)<\ep/3\}$.

On the contrary, suppose that there exists an open neighbourhood $U$ of $z$ and $i_0\in\mathbb N$ such that
\begin{equation}\label{a3}
\inf\left\{P_{q_i}^{*}\delta_y\left( K_i^{\ep/3}\right): y\in U, i\ge i_0\right\}\ge\ep/2.
\end{equation}
Clearly
\begin{equation}\label{a4}
\limsup_{T\to\infty}\sup_{\mu\in\mathcal M_{1}}\frac{1}{T}\int_{0}^{T}P_{s}^{*}\mu(U)\d s>\alpha
\end{equation}
for some $\alpha>0$.
Further, let $N\in\mathbb N$ satisfy
$
(N-i_0+1) \alpha \ep>2.
$
Choose $\gamma\in (0,\alpha\ep/2)$ such that
$$
(N-i_0+1)(\alpha \ep-2 \gamma)>2.
$$
It easily follows that there exists $T_0>0$ such that
for any $\mu\in\mathcal M_1$ and $T\ge T_0$ we have
$$
\max_{i\le N}\left|\left|\frac{1}{T}\int_{0}^{T}P_{s}^{*}\mu\,\d s - \frac{1}{T}\int_{0}^{T}P_{s+q_i}^{*}\mu\,\d s\right|\right|_{TV}<\gamma.
$$
Choose $T\ge T_0$ and $\mu\in\mathcal M_1$ such that
\begin{equation}\label{a5}
\frac{1}{T}\int_{0}^{T}P_{s}^{*}\mu(U)\d s\ge\alpha,
\end{equation}
by (\ref{a4}). From (\ref{a3}) and the Markov property it follows that
$$
P_{s+q_i}^{*}\mu \left( K_i^{\ep/3}\right)=\int_X P_{q_i}^{*}\delta_y \left( K_i^{\ep/3}\right)P_{s}^{*}(\d y)\ge
\int_U P_{q_i}^{*}\delta_y \left( K_i^{\ep/3}\right)P_{s}^{*}(\d y)\ge \frac{\ep}{2} P_{s}^{*}\mu(U)
$$
for $i\ge i_0$ and $s\ge 0$. Consequently, we have for $i_0\le i\le N$
$$
\begin{aligned}
\frac{1}{T}\int_{0}^{T} P_{s}^{*}\mu \left( K_i^{\ep/3}\right) \d s&\ge
\frac{1}{T}\int_{0}^{T} P_{s+q_i}^{*}\mu \left( K_i^{\ep/3}\right) \d s -\gamma\\
&\ge
\frac{\ep}{2}
\,\frac{1}{T}\int_{0}^{T}
P_{s}^{*}\mu(U)\d s -\gamma
\ge\frac{\ep}{2}\alpha-\gamma,
\end{aligned}
$$
by (\ref{a5}). From this and the fact that
$K_i^{\ep/3}\cap  K_j^{\ep/3}=\emptyset$ for $i\not =j$ we obtain
$$
\begin{aligned}
\frac{1}{T}\int_{0}^{T} P_{s}^{*}\mu \left(\bigcup_{i=i_0}^N K_i^{\ep/3}\right) \d s
&=\sum_{i=i_0}^N
\frac{1}{T}\int_{0}^{T} P_{s}^{*}\mu \left( K_i^{\ep/3}\right) \d s\\
&\ge
(N-i_0+1)(\ep\alpha -2\gamma)/2>1,
\end{aligned}
$$
which is impossible.

Now analogously as in the proof of Theorem 3.1 in \cite{Lasota-Szarek}, Step III, we define a sequence
of Lipschitzian functions $(f_n)_{n\ge 1}$, a sequence of points $(y_n)_{n\ge 1}$, $y_n\to z$ as $n\to \infty$, two increasing sequences of integers
$(i_n)_{n\ge 1}$, $(k_n)_{n\ge 1}$, $i_n<k_n<i_{n+1}$ for $n\in\mathbb N$, and a sequence of reals $(p_n)_{n\ge 1}$ such that
\begin{equation}\label{a6}
f_n |_{K_{i_n}}=1, \qquad 0\le f_n\le \text{\bf 1}_{K_{i_n}^{\ep/3}},\qquad\text{Lip }f_n\le 3/\ep,
\end{equation}

\begin{equation}\label{a7}
\left|P_{p_n}\left(\sum_{i=1}^nf_i\right)(z)-P_{p_n}\left(\sum_{i=1}^nf_i\right)(y_n)\right|>\frac{\ep}{4},
\end{equation}

\begin{equation}\label{a8}
P_{p_n}^{*}\delta_u\left(\bigcup_{i=k_n}^\infty K_i^{\ep/3}\right)<\frac{\ep}{16}\quad\text{for }u\in\{z,y_n\}
\end{equation}
for every $n\in\mathbb N$. From (\ref{a6})-(\ref{a8}) it follows (see the proof of Theorem 3.1 in \cite{Lasota-Szarek}, Step III, once again)
that
$$
|P_{p_n} f(z)-P_{p_n}f(y_n)|>\frac{\ep}{8}
$$
for $n\in\mathbb N$ and $f:=\sum_{n=1}^\infty f_n\in L_{b}(X)$. Since $y_n\to z$ as $n\to\infty$, this contradicts the assumption that the family
$\{P_t f: t\ge 0\}$
is equicontinuous in $z$. The proof is complete.
\end{proof}

\begin{theorem}\label{Mainthe}
Let $(P_{t})_{t\ge 0}$ be averagely bounded and concentrating at some $z\in X$. If $(P_{t})_{t\ge 0}$ satisfies the e--property, then it is asymptotically stable.
\end{theorem}

\begin{proof}
Fix $x\in X$. Since $(P_{t})_{t\ge 0}$ is averagely bounded there is $R>0$ such that
$$
\limsup_{T\to\infty}\frac{1}{T}\int_{0}^{T}P_{s}^{*}\delta_{x}(B(x, R))\d s>\frac{1}{2}.
$$
Let $(T_n)_{n\ge 1}$ be an increasing sequence of reals such that $T_n\to\infty$ as $n\to\infty$ and
$$
\frac{1}{T_n}\int_{0}^{T_n}P_{s}^{*}\delta_{x}(B(x, R))\d s >\frac{1}{2} \qquad \text{for } n\in \mathbb N.
$$
Set $\displaystyle\mu_n= \frac{1}{T_n}\int_{0}^{T_n}P_{s}^{*}\delta_{x}\,\d s$,
$n\in\mathbb N$,
and observe that there are $\mu_n^R\in\mathcal M_1^{B(x, R)}$ such that
$$
\mu_n\ge\frac{1}{2}\mu_n^R\qquad\text{for }n\in\mathbb N.
$$
Indeed, we may define $\mu_n^R$ by the formula
$
\mu_n^R={\mu_n(\,\cdot \,\cap B(x, R))}/{\mu_n(B(x, R))}\quad\text{ for $n\in\mathbb N$.}
$
Further, observe that, by concentrating at $z$, for fixed $\ep>0$ there is $\alpha>0$ such that
we have
$$
P_{s_n}^{*}\mu_n^R(B(z, \ep)) \ge \alpha
$$
for some $s_n>0$, $n\in\mathbb N$.
Hence
$$
P_{s_n}^{*}\mu_n(B(z, \ep)) \ge \frac{1}{2}\alpha
\qquad\text{for }n\in\mathbb N,
$$
by linearity of $(P_t^{*})_{t\ge 0}$. Consequently,
$$
\frac{1}{T_n}\int_0^{T_n}P_s^{*}(P_{s_n}^{*}\delta_x)(B(z, \ep))\d s \ge \frac{1}{2}\alpha
\qquad\text{for }n\in\mathbb N,
$$
and condition (\ref{inv}) in Proposition 2 is satisfied. Now Proposition 2 implies the existence of an invariant measure.
Further, from Proposition 1 it follows that for any $f\in L_{b}(X)$ and $\mu\in\mathcal M_{1}$
$$
\langle \varphi, P_{t}^{*}\mu\rangle\to \langle\varphi, \mu_{*}\rangle
$$
as $t$ tends to $+\infty$.
Application of the Alexandrov theorem finishes the proof (see \cite{Bill}).
\end{proof}

\section{The models}
\subsection{GOY and Sabra shell models and functional setting}

Let $u=(u_{-1}, u_0, u_1,\ldots)$ be an infinite sequence of complex valued functions on $[0, \infty)$ satisfying the following equations for $n=1, 2,\ldots$
\begin{equation}\label{shell_n}
\d u_{n}(t)+\nu k_{n}^{2}\nu_{n}(t)\d t+[B(u,u)]_n\d t =\sigma_{n }\d  w_{n}
\end{equation}
with the initial conditions
$$
u_{-1}(t)= u_{0}(t)=0\quad\text{and}\quad u_{n}(0)=\xi_{n}.
$$


Here $k_{n}=k_{0}2^{n}$, $k_{0}>1$ and $\nu>0$. Moreover $(w_{n}(t))_{n\ge 1}$ denotes a sequence of independent Brownian motions on some probability space $(\Omega, \mathcal F, \mathbb P)$. It is assumed that $\sigma_n\in\mathbb C$ and there is $n_0\in\mathbb N$ such that $\sigma_{n}=0$ for $n\ge n_{0}$. Further $B$ is a bilinear operator which will be defined later on.

Let $H$ be the set of all sequences
 $u=(u_1, u_2,\ldots)$ of complex numbers
such that $\sum_n |u_n|^2<\infty$. We consider $H$ as a \emph{real} Hilbert space
endowed  with the inner product $(\cdot,\cdot)$ and the norm $|\cdot|$ of the form
\begin{equation}\label{normH}
(u,v)={\rm Re}\,\sum_{n\geq 1}u_n v_n^*,\quad
|u|^2 =\sum_{n\geq 1} |u_n|^2,
\end{equation}
where $v_n^*$ denotes the complex conjugate of $v_n$. The space $H$ is separable.
Let $A:\D (A)\subset H \to H $ be the non-bounded linear operator defined by
\[
(Au)_n = k_n^2 u_n,\quad n=1,2,\ldots,\qquad \D (A)=\Big\{ u\in H\,
:\; \sum_{n\geq 1} k_n^4 |u_n|^2<\infty\Big\}.
\]
The operator $A$ is clearly self-adjoint, strictly positive definite since $(Au,u)\geq k_0^2 |u|^2$
for $u\in \D (A)$.
For any $\alpha >0$, set
$$
{\mathcal H}_\alpha = \D(A^\alpha) = \{ u\in H \, :\, \sum_{n\geq 1} k_n^{4\alpha} |u_n|^2 <+\infty\},\;
\|u\|^2_\alpha = \sum_{n\geq 1} k_n^{4\alpha} |u_n|^2 \; \mbox{\rm for }\; u\in {\mathcal H}_\alpha.
$$
Obviously  ${\mathcal H}_0=H$. Define
\[ V:=\D (A^{\frac{1}{2}}) =
 \Big\{ u\in H \, :\, \sum_{n\geq 1} k_n^2 |u_n|^2 <+\infty\Big\}\]
  and set
  \[
 \;{\mathcal H}= {\mathcal H}_{\frac{1}{4}},\, \|u\|_{\mathcal H} = \|u\|_{\frac{1}{4}}
 . \]
Then $V$
is a Hilbert space for the scalar product $(u,v)_V = {\rm Re} (\sum_n k_n^2\, u_n\, v_n^*)$, $u,v\in V$
and the associated norm is denoted by
$$
\|u\|^2 = \sum_{n\geq 1} k_n^2\, |u_n|^2.
$$
The adjoint of $V$ with respect to  the $H$ scalar product
is $V' = \{ (u_n)\in {\mathbb C}^{\mathbb N} \, :\,
\sum_{n\geq 1} k_n^{-2}\, |u_n|^2 <+\infty\}$ and $V\subset H\subset V'$
is a Gelfand triple. Let $\langle u\, ,\, v\rangle_{V', V} = {\rm Re}\left (\sum_{n\geq 1} u_n\, v_n^*\right)$ denote
the duality between $u\in V'$ and $v\in V$.

Set $ u_{-1}= u_{0}=0 $, let $a,b$ be real numbers and let
 $B : H\times V \to H$ (or  $B : V\times H \to H$) denote the bilinear operator  defined by
$$
\left[B(u,v)\right]_n= i\left( a k_{n+1} u_{n+1}^* v_{n+2}^*
+b k_{n} u_{n-1}^* v_{n+1}^* -a k_{n-1} u_{n-1}^* v_{n-2}^*
-b k_{n-1} u_{n-2}^* v_{n-1}^*
\right)
$$
for $n=1,2,\ldots$ in the  GOY shell model (see, e.g. \cite{OY89})
or
$$
\left[B(u,v)\right]_n= i\left( a k_{n+1} u_{n+1}^*\,  v_{n+2}
+b k_{n} u_{n-1}^* v_{n+1} +a k_{n-1} u_{n-1} v_{n-2}
+b k_{n-1} u_{n-2} v_{n-1}
\right),
$$
in the  Sabra shell model introduced in \cite{LPPPV98}.

Obviously, there exists $C>0$ such that
\begin{equation}\label{nee2}
|B(u, v)|\le C\|u\||v|\qquad\text{for $u\in V$ and $v\in H$}.
\end{equation}

Note that $B$ can be extended as a bilinear operator from $H\times H$ to $V'$ and that  there exists
a constant ${C}>0$ such that  given
$u,v\in H$ and $w\in V$ we have
\begin{equation}\label{trili}
|\langle B(u,v)\, ,\, w\rangle_{V', V} | + |\big( B(u,w)\, ,\, v\big) | + |\big( B(w,u)\, ,\, v\big) |
\leq {C}\, |u|\, |v|\, \|w\|.
\end{equation}
An easy computation proves that for $u,v\in H$ and $w\in V$ (resp. $v,w\in H$ and $u\in V$),
$$
\langle B(u,v)\, ,\, w\rangle_{V', V} = - \big(B(u,w)\, ,\, v\big) \; \mbox{\rm (resp.  }\,
\big( B(u,v)\, ,\, w\big)  = - \big(B(u,w)\, ,\, v\big) \mbox{\rm  \;)}.
$$
Hence $(B(v, u), u)=0$ for $u\in H$ and $v\in V$.
Furthermore, $B:V\times V\to V$ and $B : {\mathcal H}\times {\mathcal H} \to H$;
 indeed, for $u,v\in V$  (resp.  $u,v\in {\mathcal H}$)
 we have
$$
\begin{aligned}
 \|B(u,v)\|^2 & = \sum_{n\geq 1} k_n^2\, |B(u,v)_n|^2\, \leq C\, \|u\|^2 \sup_n k_n^2 |v_n|^2
\leq C\, \|u\|^2\, \|v\|^2,\\
| B(u,v)| & \leq  C\, \|u\|_{\mathcal H} \, \|v\|_{\mathcal H}.
\end{aligned}
$$
\subsection{Well-posedness}
Consider the abstract equation on $H$ of the form
\begin{equation}\label{shell}
\d u(t)=\left[-\nu Au(t)+B(u(t),u(t))\right]\d t+Q \d W(t),\quad t\geq 0
\end{equation}
with the initial condition $u(0)=\xi\in H$,
where $Q=(q_{i, j})_{i, j\in\mathbb N}$ is some matrix with $\Tr (Q Q^{*})<\infty$ and $W(t)=(w_{n}(t))_{n\ge 1}$ is a cylindrical Wiener noise on some filtered space $(\Omega, \mathcal F, (\mathcal F_{t})_{t\ge 0}, \mathbb P)$.
\begin{definition}\label{def_solution}
A stochastic process $u(t,\omega)$ is a generalized solution in $[0,T]$ of the system \eqref{shell} if
$$u(\cdot,\omega)\in C([0,T];H)\cap L^{2}(0,T;\mathcal{H})$$
for $\mathbb P$-a.e. $\omega\in \Omega$, $u$ is progressively measurable in these topologies and equation \eqref{shell} is satisfied in the integral sense
$$
\begin{aligned}
&  ( u(t),\varphi)+\int_{0}^{t}\nu(
u(s),A\varphi)\d s+\int_{0}^{t}\left(B\left(
u(s),\varphi\right)  ,u(s)\right) \d s\nonumber\\
&  =\left(\xi,\varphi\right)+\left( Q W(t),\varphi
\right)%
\end{aligned}
$$
for all $t\in [0,T]$ and $\varphi\in \D (A)$.
\end{definition}

\begin{theorem}\label{thm_solution_shell}
Let us assume that the initial condition $\xi$ is an $\mathcal{F}_{0}$-random variable with values in $H$.
Then there exists a unique solution $(u(t))_{t\ge 0}$ to equation \eqref{shell}.
Moreover, if $\mathbb E |\xi|^{2}<+\infty$, then
\begin{equation}\label{mainev}
\mathbb E |u(t)|^{2}+\int_{0}^{t}2\nu\mathbb E\|u(s)\|^{2}\d s= \mathbb E |\xi|^{2} + \Tr (Q Q^{*}) t
\end{equation}
for any $t\ge 0$.
\end{theorem}
\begin{proof}
We will prove well--posedness using a pathwise argument (for similar results see \cite{Hakima} and the references
therein). Let us introduce the Ornstein-Uhlenbeck process solution of

\begin{equation}\label{OU_process}
\left\{\begin{array}{l}
\d z(t)+\nu Az(t)\d t=Q\d W,\\
z(0)=0.\end{array}
\right.
\end{equation}
The above equation has a unique progressively measurable solution such that $\mathbb P$-a.s.
$$z\in C([0,T]; \mathcal{H})$$
(for more details see \cite{DZ}).
Set $v=u-z$. Then for $\mathbb P$-a.e. $\omega\in \Omega$

\begin{eqnarray}\label{deterministic_shell}
\left\{\begin{array}{l}
\frac{\d }{\d t} v(t)+\nu Av(t)-B(v(t)+z(t),v(t)+z(t))=0,\\
v(0)=\xi,
\end{array}
\right.
\end{eqnarray}
is a deterministic system. The existence and uniqueness of global weak solutions $v$ follow from
the Galerkin approximation procedure and then passing to the limit using the appropriate compactness theorems.
We omit the details which can be found in \cite{Hakima} and the references therein. Instead, we present the formal
computations which lead to the basic a priori estimates, this is in order to stress the role played by $z$.
Using equation \eqref{deterministic_shell} and various properties of the nonlinear operator $B$, we have
\begin{eqnarray*}
\frac{1}{2}\frac{\d}{\d t}|v(t)|^{2}+\nu\|v(t)\|^{2}&\leq& |(B(v(t)+z(t), z(t)),v(t))|\\
 &\leq& C\|v(t)\||v(t)+z(t)||z(t)|\\
 &\leq& \frac{\nu}{2}\|v(t)\|^{2}+C(\nu)\left(|v(t)|^{2}|z(t)|^{2}+|z(t)|^{4}\right).
 \nonumber
\end{eqnarray*}
Using Gronwall's Lemma and the fact that $\|z\|_{C([0,T]; \mathcal{H})}\leq C(\omega)$, we have
$$
\sup_{0\leq t\leq T}|v(t)|^{2}\leq C(|\xi|, T, C(\omega)).
$$
Again, using the above inequality in the previous estimate, we obtain that
$$
\int_{0}^{T}\|v(s)\|^{2}\d s\leq C(|\xi|, T, C(\omega)).
$$
Then, by classical arguments, see \cite{Temam}, $v\in C([0,T]; H)\cap L^{2}(0,T; \D (A^{1/2}))$.
Therefore $u=v+z\in C([0,T]; H)\cap L^{2}(0,T; \D (A^{1/4}))$ $\mathbb P$-a.s.

To finish the proof observe that condition (\ref{mainev}) follows from
It$\hat{\text{o}}$'s formula.
\end{proof}
The uniqueness of solutions is established in the following theorem.

\begin{theorem}\label{continuous_dependence}
Let $\left(  u^{\left(  1\right)  }(t)\right)  _{t\geq0}%
$, $\left(  u^{\left(  2\right)  }(t)\right)  _{t\geq0}$, be two continuous
adapted solutions of $(\ref{shell})$ in $H$, with the initial conditions
$u_{0}^{\left(  1\right)  }$ and $u_{0}^{\left(  2\right)  }$ as above. Then
there is a constant $C(\nu)>0$, depending only on $\nu$, such that $\mathbb P$-a.s.

\begin{equation*}
\left|u^{1}(t)-u^{2}(t)\right|^{2}
\leq {\rm e}^{C(\nu)\int_{0}^{t}\left|  u^{1}(s)\right|^{2}ds}
\left|u_{0}^{1}-u_{0}^{2}\right|^{2}\quad t\geq 0.
\end{equation*}
\end{theorem}
\begin{proof}
Let us put $u(t)=u^{1}(t)-u^{2}(t)$.
Then $u$ is the solution of the following equation

$$
\d u+\nu Au \d t-\left( B(u^{1},u^{1})-B(u^{2},u^{2})\right) \d t=0.
$$
Using again the properties of operator $B$, we obtain

\begin{eqnarray*}
\frac{\d}{\d t}|u|^{2}+\nu\|u\|^{2}&\leq& |(B(u,u^{1}),u)|\\
&\leq&\frac{\nu}{2}\|u\|^{2}+C(\nu)|u|^{2}|u^{1}|^{2}.
\end{eqnarray*}
Hence, by the Gronwall lemma, we obtain that

\begin{equation*}
|u(t)|^{2}\leq |u(0)|^{2}{\rm e}
^{C(\nu)\left(\int_{0}^{T}|u^{1}(s)|^{2}ds\right)},
\end{equation*}
which finishes the proof.
\end{proof}

\section{Stability of the model}

Let a diagonal matrix $Q=(q_{i, j})_{i, j\in\mathbb N}$ be such that there is $n_0\in\mathbb N$ and $q_{n, n}=0$ for $n\ge n_0$. Consider the equation on $H$ of the form
\begin{equation}\label{Eq1}
\d u (t)=[-\nu Au(t)+B(u(t), u(t))]\d t +Q\d W(t) \quad t\ge 0,
\end{equation}
where $(W(t))_{t\ge 0}$ is a certain cylindrical Wiener process on a filtered space $(\Omega, \mathcal F, (\mathcal F_{t})_{t\ge 0}, \mathbb P)$.
\vskip5mm

By Theorem  \ref{thm_solution_shell} for every $x\in H$ there is a unique continuous solution $(u^{x}(t))_{t\ge 0}$ in $H$, { hence the transition semigroup is well defined. From Theorem \ref{continuous_dependence}} we obtain that the solution satisfies the Feller property, i.e. for any $t\ge 0$ if $x_{n}\to x$ in $H$, then $\mathbb E f(u^{x_{n}}(t))\to \mathbb E f(u^{x}(t))$ for any $f\in C_{b}(H)$.
Set
$$
P_{t}f(x)=\mathbb E f(u^{x}(t))\quad\text{for any $f\in C_{b}(H)$}.
$$
Obviously $(P_{t})_{t\ge 0}$ is stochastically continuous.
First note that $D P_{t} f(x)[v]$, the value  of the Frechet derivative $D P_{t} f(x)$ at $v\in H$, is
equal to $ \E \,\left\{ D{f}(u^x(t)) [U(t)]\right\}$, where
$U(t):=\partial u^x(t)[v]$ and
$$
\partial u^x(t)[v] := \lim_{\eta \downarrow 0}
\frac{1}{\eta} \left( u^{x+\eta v}(t)- u^x(t)\right)
$$
and the limit is in $L^2(\Omega,\cF,\P; H)$ (see \cite{KPS10} also \cite{HM06}). The process
$U=\left(U(t)\right)_{t\ge0}$ satisfies the linear evolution
equation
\begin{equation}
\label{ET27}
\begin{aligned}
\frac{\d U(t)}{\d t}&= -\nu A U(t)+ B(u^x(t), U(t)) + B(U(t), u^x(t)),\\
U(0)&=v.
\end{aligned}
\end{equation}

 Suppose that $\cX$ is a certain Hilbert space and $\Phi\colon
H\to \cX$ a Borel measurable function. Given an
$(\cF_t)_{t\ge0}$-adapted process $g\colon [0,\infty) \times \Omega \to
H$ satisfying $ \E \int_0^t |g(s)|^2\d s <\infty $ for each
$t\ge0$ we denote by $\cD_g\Phi(u^x(t))$  the Malliavin derivative
of $\Phi(u^x(t))$ in the direction of $g$; that is the
$L^2(\Omega,\cF,\P;\cX)$-limit, if exists, of
$$
\cD_g\Phi(u^x(t)):=\lim_{\eta \downarrow
0}\frac{1}{\eta} \left[ \Phi(u^x_{\eta g}(t)) -\Phi(
u^x(t))\right],
$$
where $u^x _{g}(t)$, $t\ge 0$, solves the equation
$$
\d u^x_{g}(t) = \left[-\nu Au^x_{ g}(t)+ B(u^x_{g}(t), u^x_{g}(t))\right] \d t +
Q \left( \d W(t) + g(t)\d t\right),\qquad u^x_{g}(0)= x.
$$
In particular, one can easily show that when $\cX=H$ and $\Phi=I$, where $I$ is the identity operator,
the Malliavin derivative of $u^x(t)$ exists and
 the process $D(t):= \cD_gu^x(t)$, $t\ge 0$, solves the linear equation
\begin{equation}\label{ET28}
\begin{aligned}
\frac{\d D}{\d t}(t)&=-\nu A D(t) + B(u^x(t),D(t))
+B(D(t),u^x(t)) + Q g(t),\\
&\\
D(0)&=0.
\end{aligned}
\end{equation}

 Directly from the definition of the Malliavin
derivative we conclude the \emph{chain rule:} suppose that $\Phi\in
C^1_b(H; \cX)$ then
$$
\cD_g\Phi( u^x(t))=D\Phi(u^x(t))[D(t)].
$$
(Here $C^1_b(H; \cX)$ denotes the space of all bounded continuous functions $\Phi :H\to\cX$ with continuous and bounded first derivative with the natural norm. In the case when $\cX=\mathbb R$ we simply write $C^1_b(H)$.)
In addition, the  \emph{integration by parts formula} holds, see
Lemma 1.2.1, p. 25 of \cite{Nualart}. Indeed, suppose that $\Phi\in
C^1_b(H)$. Then
\begin{equation}
\label{083003} \E[\cD_g\Phi( u^x(t))]=
\E\left[\Phi(u^x(t))\int_0^t (g(s), \d
W(s))\right].
 \end{equation}

\begin{lemma}\label{lemma_1}
Let $\eta\in (0, \nu/(2\max q^2_{i, i})]$. Then we have
$$
\mathbb E(\exp(\eta |u^x(t)|^2+\eta\nu\int_0^t\|u^x(s)\|^2\d s))\le 2\exp(\eta  (\Tr Q^2) t+\eta |x|^2).
$$
\end{lemma}

\begin{proof} Fix  $\eta\in (0, \nu/(2\max q^2_{i, i})]$. Let $M(t)=\eta\int_{0}^{t}(u^{x}(s), Q\d W(s))$ and let $N(t)=M(t)-\eta\nu\int_{0}^{t}\|u^{x}(s)\|^2\d s.$
Set  $\alpha=\nu/\max q^2_{i, i}$. Then we have $\nu\|u^{x}(s)\|^{2}\ge\alpha |Qu^{x}(s)|^{2}.$ Now observe that
$N(t)\le M(t)-(\alpha/\eta)\langle M\rangle(t)$, where $\langle M\rangle(t)$ denotes the quadratic variation of the continous $L^{2}$--martingale $M$ with the filtration generated by the noise. Hence by a standard variation of the  Kolmogorov--Doob martingale inequality (see \cite{RY}) we have
$$
\mathbb P(N(t)\ge K)\le\exp (-\alpha K/\eta)
$$
and consequently we obtain
$$
\mathbb P(\exp N(t)\ge \exp K)\le\exp (-\alpha K/\eta)\le\exp(-2K)
$$
for any $K>0$. An easy observation that if some positive random variable, say $X$, satisfies the condition $\mathbb P(X\ge C)\le C^{-2}$ for every $C>0$, then $\mathbb E X\le 2$ gives
$$
\mathbb E (\exp(\eta|u^{x}(t)|^{2}+\eta\nu\int_{0}^{t}\|u^{x}(s)\|^2\d s-\eta(\Tr Q^{2})t-\eta|x|^{2}))\le2,
$$
by It$\hat{\text{o}}$'s formula.
This completes the proof.
\end{proof}

 The crucial role in our consideration is played by the following lemma. The idea of its proof is taken from \cite{HM06}.

\begin{lemma}\label{E-property} Let $(P_{t})_{t\ge 0}$  correspond to problem $(\ref{Eq1})$. If $Q$ satisfies the condition:
\begin{equation}\label{Eq3.04.11}
q_{1, 1}, \ldots, q_{N_*, N_*}\neq 0 \quad\text{for $N_*> \log_2(2C^2\max q^2_{i, i} /\nu^3 +\Tr Q^2/(2\max q^2_{i, i} ))/2$},
\end{equation}
where $C>0$ is given by (\ref{nee2}), then for any $f\in C^1_b(H)$ and $R>0$ there exists a constant $C_{0}>0$ such that
\begin{equation}\label{E_10.11}
\sup_{t\ge0}\sup_{|x|\le R}\sup_{|v|\le 1}|D P_{t} f(x)[v]|\le C_{0} \|f\|_{C_b^1({H})}.
\end{equation}
\end{lemma}

\begin{proof}  Fix $N_*> \log_2(2C^2\max q^2_{i, i} /\nu^3 +\Tr Q^2/(2\max q^2_{i, i} ))/2$. The proof will be split into three steps.

{\bf Step I:} Let $g: [0, \infty)\times\Omega\to H$ be a measurable function such that  $ \E \int_0^t |g(s)|^2\d s <\infty $ for any
$t\ge0$.
Let
$\om_t(x):={\cD}_{g}u^x(t)$ and $ \rho_t(v,x):=\partial u^x(t)[v]-
{\cD}_{g}u^x(t)$.  Then,
$$
\begin{aligned}
D P_{t} f(x)[v]&\ = \ \  \E\left\{\,Df(u^x(t))[\om_t(x)]\right\}
+\E\,\left\{ D
f(u^x(t))[\rho_t(v,x)]\right\}\\
&\ =\ \ \E\, \left\{{\cD}_{g} f(u^x(t))\right\} +\E\,\left\{
 Df(u^x(t))[\rho_t(v,x)]\right\}
\\
&\stackrel{\eqref{083003}}{=} \E\,
\left\{f(u^x (t))\int_0^t ( g(s),\d
W(s))\right\}+ \E\,\left\{
 Df(u^x(t))[\rho_t(v,x)]\right\}.
\end{aligned}
$$

We have
$$
\left| \E\, \left\{f(u^x(t))\int_0^t ( g(s), \d
W(s))\right\}\right|\le \|f\|_{L^\infty}\left(\E\,
\int_0^t |g(s)|^2\d s\right)^{1/2}
$$
and
$$
\left|\E\,\left\{
 Df(u^x(t))[\rho_t(v,x)]\right\} \right|\le \|f\|_{C_b^1({H})}\E\,
|\rho_t(v,x)|\le \|f\|_{C_b^1({H})}(\E\,
|\rho_t(v,x)|^{2})^{1/2}.
$$

{\bf Step II:} Let $\xi(t)=(\xi_1(t), \xi_2(t),\ldots): [0, \infty)\to H$ be a solution to the following system:
$$
\begin{aligned}
&\frac{\d \xi_i(t)}{\d t}=-\frac{\xi_i(t)}{2\sqrt{\sum_{i=1}^{N_*} \xi_i^2(t)}}\qquad\text{for $ i=1, \ldots, N_*$}\\
&\frac{\d \xi_i(t)}{\d t}=-\nu k_i^2\xi_i(t) +[B(u^x(t), \xi(t))+B(\xi(t), u^x(t))]_i\qquad\text{ for $i\ge N_*+1$}.
\end{aligned}
$$
with $\xi(0)=v$. We assume also that ${\xi_i(t)}/{2\sqrt{\sum_{i=1}^{N_*} \xi_i^2(t)}}=0$ if  $\sqrt{\sum_{i=1}^{N_*} \xi_i^2(t)}=0$ (see \cite{HM06}).
Observe that $\xi_1(t), \xi_2(t),\ldots, \xi_{N_*}(t)=0$ for $t\ge 2$.

Now we choose $g: [0, +\infty)\times \Omega\to H$ to be given by the formulae:
$$
g_i(t)=\frac{1}{q_{i, i}}\left(-\nu k_i^2\xi_i(t) + [B(u^x(t), \xi(t))+B(\xi(t), u^x(t))]_i-\frac{\xi_i(t)}{2\sqrt{\sum_{i=1}^{N_*} \xi_i^2(t)}}\right)
$$
for $i=1, \ldots, N_*$ and $g_i(t)=0$ for $i\ge N_*+1$.

It is easy to see that $\rho_t=\xi(t)$ for any $t\ge 0$. Indeed, observe that
$$
\frac{\d \xi(t)}{\d t} +Qg(t)=-\nu A\xi(t)+B(u^x(t), \xi(t))+B(\xi(t), u^x(t))
$$
and
$$
\xi(0)=v.
$$
On the other hand, subtracting equation (\ref{ET27}) from (\ref{ET28}) we obtain the equation for $\rho_t$.
Since $\rho_t$ and $\xi(t)$ solve the same equation with the same initial condition $\rho_0=\xi(0)=v$, we obtain $\rho_t=\xi(t)$ for $t\ge 0$.

{\bf Step III:} To show (\ref{E_10.11}) it is enough to prove that
$$
\sup_{|x|\le R}\sup_{|v|\le 1}\mathbb E\int_0^{\infty} |g(s)|^2\d s<\infty
$$
and
$$
\sup_{t\ge 0}\sup_{|x|\le R}\sup_{|v|\le 1}\mathbb E |\xi(t)|^2<\infty.
$$

We know that $\sum_{i=1}^{N_*} |\xi_i(t)|^2\le |v|^2\le 1$ for $t\ge 0$. In particular $\xi_i(t)=0$ for $t\ge 2$ and $i=1, \ldots, N_*$.  Let $\zeta(t)=(\xi_{N_*+1}(t), \xi_{N_*+2}(t),\ldots)$. It is easy to see that $\zeta$ satisfies the inequality
\begin{equation}\label{e11.3.1}
\frac{\d |\zeta(t)|^2}{\d t}\le -\nu k_{N_*}^2|\zeta(t)|^2 +2C\|u^x(t)\||\zeta(t)|^2+2\tilde C\|u^x(t)\||\zeta(t)|\quad\text{for $t\ge 0$},
\end{equation}
where $\tilde C$ is some positive constant dependent only on $C$.
Choose $\varepsilon>0$ and $\gamma\in (0, 1)$ such that
$$
-\nu k_{N_*}^2 +\varepsilon+2C^2 \max q_{i, i}^2/\nu^2+\nu \Tr Q^2   /(2\gamma\max q_{i, i}^2)<0.
$$
From equation \REQ{e11.3.1} we derive
$$
\frac{\d |\zeta(t)|^2}{\d t}\le(-\nu k_{N_*}^2 +2C\|u^x(t)\|+\varepsilon)|\zeta(t)|^2+C(\varepsilon) \|u^x(t)\|^2
$$
and using Gronwall's lemma we obtain
$$
\begin{aligned}
|\zeta(t)|^2&\le \left(|v|^2+ C(\varepsilon)\int_0^t \|u^x(s)\|^2\d s\right)e^{(-\nu k_{N_*}^2 +\varepsilon)t+2C\int_0^t\|u^x(s)\|\d s}\\
&\le e^{(-\nu k_{N_*}^2 +\varepsilon)t}\left[1+C(\varepsilon)\int_0^t \|u^x(s)\|^2\d s\right] e^{2C\int_0^t\|u^x(s)\|\d s}.
\end{aligned}
$$
Hence we obtain that there exist constant $A>0$ (independent of $t\ge 0$, $v\in B(0, 1)$ and $x\in B(0, R)$) such that
$$
|\zeta(t)|^2\le A \exp(\gamma(-\nu k_{N_*}^2 +\varepsilon +2C^2\max q_{i, i}^2/\nu^2)t) \exp\left({\nu/(2\max q_{i, i}^2)\int_0^t\|u^x(s)\|^2\d s}\right)
$$
for all $t\ge 0$, by the fact that $-\nu k_{N_*}^2 +\varepsilon +2C^2\max q_{i, i}^2/\nu^2<0$.
Thus
$$
\begin{aligned}
&\sup_{|x|\le R, |v|\le 1, t\ge 0} \mathbb E |\zeta(t)|^2\\
&\le A \exp(\gamma(-\nu k_{N_*}^2 +\varepsilon +2C^2\max q_{i, i}^2/\nu^2)t)\mathbb E\left(\exp\left({\nu/(2\max q_{i, i}^2)\int_0^t\|u^x(s)\|^2\d s}\right)\right).
\end{aligned}
$$
Using Lemma \ref{lemma_1} we obtain
$$
\sup_{t\ge 0, |x|\le R, |v|\le 1} \mathbb E |\zeta(t)|^2\le \tilde A \exp(\gamma(-\nu k_{N_*}^2 +\varepsilon +2C^2\max q_{i, i}^2/\nu^2 + \nu/(2\gamma\max q^2_{i, i})\Tr Q^2)t)
$$
for some $\tilde A>0$. On the other hand, by the definition of $N_*$, $k_n$ and the choice of $\varepsilon, \gamma$ we have
$$
-\nu k_{N_*}^2 +\varepsilon +2C^2\max q_{i, i}^2/\nu^2 + \nu/(2\gamma\max q^2_{i, i})\Tr Q^2<0.
$$
Now we must evaluate
$$
\mathbb E\int_0^t|g(s)|^2\d s\le 2\sup_{0\le s\le 2} \mathbb E|g(s)|^2+\mathbb E\int_2^t|g(s)|^2\d s.
$$
The first term on the right side of the above inequality is bounded uniformly in $|x|\le R$ and $|v|\le 1$. Further, for $s\ge 2$ we have
$$
|g(s)|\le \tilde C\|u^x(s)\||\zeta(s) |
$$
and
$$
\begin{aligned}
&\mathbb E\int_2^t|g(s)|^2\d s\le\tilde C^2\mathbb E\int_2^{\infty} \|u^x(s)\|^2|\zeta(s)|^2\d s\\
&\le \hat C\,\mathbb E[\int_2^{\infty}\|u^x(s)\|^2 \exp(\gamma(-\nu k_{N_*}^2 +\varepsilon +2C^2\max q_{i, i}^2/\nu^2)s)\\
&\times \exp (\nu/(2\max q^2_{i, i})\int_0^s\|u^x(r)\|^2\d r)\d s]\\
&\le \hat C\mathbb E[\int_2^{\infty}\exp(\gamma(-\nu k_{N_*}^2 +\varepsilon +2C^2\max q_{i, i}^2/\nu^2)s)\\
&\times\exp(\nu/(2\max q^2_{i, i})|u^x(s)|^2+\nu/(2\max q^2_{i, i})\int_0^s\|u^x(r)\|^2\d r)\d s]\\
&\le\hat C\int_2^{\infty}[\exp(\gamma(-\nu k_{N_*}^2 +\varepsilon +2C^2\max q_{i, i}^2/\nu^2)s)\\
&\times\mathbb E\exp(\nu/(2\max q^2_{i, i})|u^x(s)|^2+\nu/(2\max q^2_{i, i})\int_0^s\|u^x(r)\|^2\d r)]\d s\\
&\le { C'}\int_2^{\infty}\exp(\gamma(-\nu k_{N_*}^2 +\varepsilon +2C^2\max q_{i, i}^2/\nu^2+ \nu\Tr Q^2 /(2\gamma\max q^2_{i, i}))s)\d s,
\end{aligned}
$$
for any $x\in B(0, R)$, where the constant $C'$ depends only on $R$. Using again the assumption on $N_*$ we obtain
$$
\sup_{|x|\le R, |v|\le 1}\mathbb E\int_2^{\infty}|g(s)|^2\d s<\infty.
$$
This completes the proof.
\end{proof}

\begin{lemma}\label{AvB} (Average boundedness) Let $(P_{t})_{t\ge 0}$ correspond to problem $(\ref{shell})$. Then $(P_{t})_{t\ge 0}$
is averagely bounded.
\end{lemma}

\begin{proof} Fix an $\ep>0$ and let $r>0$ be given. If $x\in B(0, r)$, then
$$
\begin{aligned}
&\frac{1}{T}\int_{0}^{T} P_{s}^{*}\delta_{x}(H\setminus B(0, R))\d s
=\frac{1}{T} \int_{0}^{T}\mathbb P(|u^{x}(s)|>R)\d s
\le \frac{1}{T} \int_{0}^{T}\mathbb P(\|u^{x}(s)\|>R)\d s\\
&=\frac{1}{T}\int_{0}^{T}\mathbb P(\|u^{x}(s)\|^{2}>R^{2})\d s\le \frac{1}{T}\int_{0}^{T}\frac{\mathbb E\|u^{x}(s)\|^2}{R^{2}}\d s\\
&=\frac{1}{\nu R^{2}}\frac{1}{T}\int_{0}^{T}\nu \mathbb E\|u^{x}(s)\|^{2}\d s\le \frac{1}{\nu R^{2}}
(\Tr Q^2+|x|^{2}/T)\le \frac{1}{\nu R^{2}} (\Tr Q^2+r^{2}/T)
\end{aligned}
$$
for arbitrary $R>0$, by (\ref{mainev}). Hence there is $R_{0}>0$ such that
$$
\liminf_{T\to+\infty}\frac{1}{T}\int_{0}^{T} P_{s}^{*}\delta_{x}(B(0, R_{0}))\d s>1-\ep.
$$
On the other hand, by Fatou's lemma we have
$$
\begin{aligned}
\liminf_{T\to+\infty}\frac{1}{T}\int_{0}^{T} P_{s}^{*}\mu(B(0, R_{0}))\d s&\ge
\int_{H}\left(
\liminf_{T\to+\infty}\frac{1}{T}\int_{0}^{T} P_{s}^{*}\delta_{x}(B(0, R_{0}))\d s\right)\mu(\d x)\\
&\ge \int_{H}(1-\ep)\mu(\d x)=1-\ep
\end{aligned}
$$
for any $\mu\in\mathcal M_1^{B(0, r)}$. The proof is complete.
\end{proof}

\begin{lemma}\label{Con} (Concentrating at $0$) Let $(P_{t})_{t\ge 0}$ correspond to problem $(\ref{shell})$. Then $(P_{t})_{t\ge 0}$ is concentrating at $0$.
\end{lemma}

\begin{proof} Consider first the deterministic equation
$$
\d v^{x}(t)=[-\nu A v^{x}(t)+B(v^{x}(t), v^{x}(t))]\d t
$$
with the initial condition $v^{x}(0)=x$. Then
$$
\frac{1}{2} \frac{\d |v^{x}(t)|^{2}}{\d t}\le -\nu k_{0}|v^{x}(t)|^{2}
$$
and consequently
$$
|v^{x}(t)|^{2}\to 0\qquad\text{as $t\to+\infty$}
$$
uniformly on bounded sets. Further, fix $\ep>0$ and $r>0$. Let $t_{0}>0$ be such that $v^{x}(t_{0})\in B(0, \ep/2)$ for all $x\in B(0, r)$.
We may show (see Theorem 8 in \cite{Hakima}) that the process corresponding to the considered model is stochastically stable (see also \cite{KPS10}), i.e. there exists $\eta>0$ and the set $F_{\eta}=\{\omega\in \Omega: \sup_{ 0\le t\le t_{0}}|QW(t)(\omega)|\le\eta\}$ such that
$$
|u^{x}(t_{0})(\omega)-v^{x}(t_{0})|\le\ep/2\qquad\text{for any $\omega\in F_{\eta}$.}
$$
Since the process is degenerate, we have $\alpha: =\mathbb P(F_{\eta})>0$.
Consequently, we obtain
$$
P_{t_0}^{*}\delta_{x}(B(0, \ep))\ge\mathbb P(\{\omega\in\Omega: u^{x}(t_0)(\omega)\in B(0, \ep)\})\ge \mathbb P(F_{\eta})=\alpha
$$
for arbitrary $x\in B(0, r).$ Since
$$
P_{t_0}^{*}\mu (B(0, \ep))=\int_{H} P_{t_0}^{*}\delta_{x}(B(0, \ep))\mu(\d x),
$$
we obtain $P_{t_0}^{*}\mu (B(0, \ep))\ge\alpha$ for any $\mu\in\mathcal M_{1}^{B(0, r)}$. But $\ep>0$ and $r>0$ were arbitrary and hence the concentrating property follows.
\end{proof}

We may formulate the main theorem of this part of our paper.

\begin{theorem} The semigroup $(P_{t})_{t\ge 0}$ corresponding to problem $(\ref{shell})$
with $Q$ satisfying condition (\ref{Eq3.04.11})
is asymptotically stable. In particular, it admits a unique invariant measure.
\end{theorem}

\begin{proof} From Lemma \ref{E-property} it follows that the semigroup $(P_{t})_{t\ge 0}$ satisfies the e--property. It is also averagely bounded and concentrating at $0$, by Lemmas \ref{AvB} and \ref{Con}. Application of Theorem  \ref{Mainthe} finishes the proof.
\end{proof}

{\bf{Remark:}} Observe that condition (\ref{Eq3.04.11}) implies that the system with not too much noise is stable even when the noise is added to the first mode only.


\begin{thebibliography}{99}

\bibitem{ABCPV01} Arad, L., Biferale, L., Celani, A., Procaccia, I., Vergassola, M., {\em Statistical conservation laws in turbulent transport}, Phy. Rev. Lett., {\bf 87} (2001), 164-502.

\bibitem{Bill} Billingsley, P., {\em Convergence of
Probability Measures}, John Willey, New York (1968).

\bibitem{BKL} Bricmont, J., Kupiainen, A. and Lefevere, R., {\em Exponential mixing of the 2D stochastic Navier-Stokes dynamics}, Comm. Math. Phys. {\bf 230} (1) (2002),  87–132.

\bibitem{Hakima} Barbato, D.,  Barsanti, M., Bessaih, H. and Flandoli, F., {\em Some rigorous results on a stochastic Goy model}, Journal of Statistical Physics, {\bf 125} (3) (2006), 677-716.

\bibitem{Biferale03} Biferale, L., {\em Shell models of energy cascade in turbulence}, Annu. Rev. Fluid. Mech., {\bf 35} (2003), 441-468.

\bibitem{CLT06} Constantin, P., Levant, B. and Titi, E. S., {\em Analytic study of the shell model of turbulence}, Physica D,  {\bf 219} (2006), 120--141.

\bibitem{DZ} Da Prato, G. and Zabczyk, J., {\em Stochastic equations in infinite dimensions}, Cambridge University Press, Cambridge (1992).

\bibitem{Ferrario99} Ferrario, B., {\em Stochastic Navier-Stokes equations: analysis of the noise to have a unique invariant measure},  Ann. Mat. Pura Appl., {\bf 4} 177 (1999), 331-347.

\bibitem {FlaCIME} Flandoli, F., {\em An Introduction to 3D Stochastic Fluid
Dynamics}, CIME Lecture Notes (2005).

\bibitem{FM95} Flandoli, F. and Maslowski, B., {\em Ergodicity of the 2-D Navier-Stokes Equation Under Random Perturbations}, Comm. Math. Phys. {\bf 171} (1995), 119-141.

\bibitem{Frisch} Frisch, U., {\em Turbulence,} Cambridge University Press, Cambridge (1995).

\bibitem {Gallavotti} Gallavotti, G., {\em Foundations of Fluid Dynamics}, Texts and
Monographs in Physics, Springer-Verlag, Berlin (2002). Translated from Italian.



\bibitem{Hai} Hairer, M., {\em Exponential mixing properties of stochastic PDEs through asymptotic coupling}, Probab. Theory Related Fields {\bf 124} (3) (2002), 345–380.



\bibitem{HM06} Hairer, M. and Mattingly, J., {\em  Ergodicity of the 2D
Navier-Stokes equations with degenerate stochastic forcing},  Ann.
of  Math. {\bf 164}  (2006), 993--1032.

\bibitem{KPS10} Komorowski, T., Peszat, S. and Szarek, T., {\em  On ergodicity of some Markov processes}, Ann. of Prob. {\bf 38} (4)
(2010), 1401-1443.

\bibitem{Ku-Sh} Kuksin, S. and  Shirikyan, A.,{\em Stochastic dissipative PDEs and Gibbs measures}, Comm. Math. Phys. {\bf 213} (2)  (2000), 291–330.

\bibitem{LasMackey} Lasota, A. and Mackey, M.C., {\em Chaos, fractals, and noise. Stochastic aspects of dynamics}, Springer-Verlag, New York (1994).

\bibitem{Lasota-Szarek}
Lasota, A. and Szarek, T., {\em Lower bound technique in the theory
of a stochastric differential equation}, J. Differential Equations
{\bf 231}  (2006), 513--533.

\bibitem{LasYorke}
Lasota, A. and Yorke, J., {\em On the existence of invariant measures for piecewise monotonic transformations}, Trans. Amer. Math. Soc. {\bf 186} (1973), 481-488.

\bibitem{LasYorke2} Lasota, A. and Yorke, J., {\em Lower bound technique for Markov operators and iterated function systems},. Random Comput. Dynam. {\bf 2} (1994), 41-77.

\bibitem{LPPPV98}
Lvov, V.S., Podivilov, E., Pomyalov, A., Procaccia, I. and Vandembroucq, D.,
{\em Improved shell model of turbulence},
{Physical Review E} {\bf 58} (1998), 1811--1822.

\bibitem{Nualart} Nualart, D., {\em  The Malliavin Calculus and Related
Topics}, Springer-Verlag,  Berlin, Heidelberg, New York (1995).



\bibitem{RY} Revuz, D. and Yor, M., {\em Continuous Martingales and Brownian Motion}, Springer-Verlag, Berlin (1994).

\bibitem{Szar} Szarek, T., {\em The stability of Markov operators on Polish spaces}, Studia Math. {\bf 143} (2)  (2000), 145–152.

\bibitem{SzarDiss} Szarek, T., {\em Invariant measures for nonexpensive Markov operators on Polish spaces}, Dissertationes Math. (Rozprawy Mat.) {\bf 415} (2003), 62 pp.


\bibitem{Trap} Traple, J., {\em On the asymptotic stability of Markov semigroups}, Bull. Polish Acad. Sci. Math. {\bf 44} (2)  (1996), 183–195.
\bibitem{OY89}
Ohkitani, K. and Yamada, M., {\em Temporal intermittency in the energy cascade process
and local Lyapunov analysis in fully developed model of turbulence},
{Prog. Theor. Phys.}
{\bf 89} (1989), 329--341.

\bibitem{Temam} Temam, R., {\em Navier-Stokes Equations, Theory and Numerical
Analysis}, 3rd ed. North-Holland, Amsterdam (1984).



\bibitem{Worm} Worm, D., {\em Semigroups on spaces of measures}, Thomas Stieltjes Institute for Mathematics Ph.D. Thesis, University of Leiden, Leiden (2010).
\end{thebibliography}
\end{document}